\documentclass[journal,comsoc]{IEEEtran}
\usepackage[T1]{fontenc}
\usepackage{cite}
\usepackage{amsmath,xparse,amsfonts,mathrsfs}
\usepackage{algorithmic}
\usepackage{bm}
\usepackage{graphicx}
\usepackage{textcomp}
\usepackage{xcolor}
\usepackage[utf8]{inputenc}
\usepackage[english]{babel}

\usepackage{amsthm}

\newtheorem{prop}{Proposition}
\newtheorem{assumption}{Assumption}

\newtheorem{remark}{Remark}

\usepackage[cmintegrals]{newtxmath}

\begin{document}
\renewcommand{\figurename}{Fig.}
\title{A Statistical Characterization of Localization Performance in Millimeter-Wave Cellular Networks}

\author{Jiajun He,~\IEEEmembership{Student Member,~IEEE,}
        Young Jin Chun,~\IEEEmembership{Member,~IEEE}
\thanks{This work was supported by in part by the City University of Hong Kong (CityU), Startup Grant 7200618, and in part by the CityU, Strategic Research Grant 21219520}        
\thanks{J. He and Y. J. Chun are with the Department of Electrical Engineering, City University of Hong Kong, Hong Kong, China (e-mail: jiajunhe5-c@cityu.edu.hk; yjchun@cityu.edu.hk)}}
        

\maketitle

\begin{abstract}
Millimeter-wave (mmWave) communication is a promising solution for achieving high data rate and low latency in 5G wireless cellular networks. Since directional beamforming and antenna arrays are exploited in the mmWave networks, accurate angle-of-arrival (AOA) information can be obtained and utilized for localization purposes. The performance of a localization system is typically assessed by the Cram$\bf \acute{\rm e}$r-Rao lower bound (CRLB) evaluated based on fixed node locations. However, this strategy only produces a fixed value for the CRLB specific to the scenario of interest. To allow randomly distributed nodes,  stochastic geometry has been proposed to study the CRLB for time-of-arrival-based localization. To the best of our knowledge, this methodology has not yet been investigated for AOA-based localization. In this work, we are motivated to consider the mmWave cellular network and derive the CRLB for AOA-based localization and its distribution using stochastic geometry. We analyze how the CRLB is affected by the node locations' spatial distribution, including the target and participating base stations. To apply the CRLB on a network setting with random node locations, we propose an accurate approximation of the CRLB using the $\lceil L/4 \rceil$-th value of ordered distances where $L$ is the number of participating base stations. Furthermore, we derive the localizability of mmWave network, which is the probability that a target is localizable, and examine how the network parameters influence the localization performance. These findings provide us deep insight into optimum network design that meets specified localization requirements.
\end{abstract}


\IEEEpeerreviewmaketitle

\section{Introduction}
\IEEEPARstart{D}{ue} to the emergence of internet-of-things, positioning techniques have received considerable attention, which can be utilized to enhance user experience of location-based services, including navigation, mapping, and intelligent transportation systems \cite{ref1}. Fifth-generation (5G) wireless network access interface together with its large bandwidth, high carrier frequency, and massive antenna array offers excellent opportunities for accurate localization, and millimeter-wave (mmWave) is a promising technology for the 5G wireless communication systems to meet such requirements. 
Wireless networks enable us to obtain accurate location-bearing information from estimating the channel parameters, such as  time-of-arrival (TOA), time-difference-of-arrival (TDOA), received signal strength (RSS), and angle-of-arrival (AOA). In mmWave networks, we can exploit the large antenna array and highly directional transmission to acquire the AOAs with high precision\cite{ref2}. Large-scale directional antenna arrays are leveraged due to the small wavelength of mmWave signals, which can generate highly directional beams and provide large beamforming gain\cite{ref3}. In this paper, we analyze the localization performance of the mmWave wireless network using the AOA measurements.

A target is \textit{localizable} if its position can be determined without ambiguity with a sufficient number of participating base stations (BSs). The AOA-based positioning requires at least $2$ BSs to determine the location of the target in a two-dimensional (2-D) plane \cite{ref4}. Since the number of participating BSs determines the accuracy of the localization, we introduce the notion of $L$-localizability, which indicates the probability of at least $L$ BSs participating in the localization procedure. 

Furthermore, Cram$\bf \acute{\rm e}$r-Rao lower bound (CRLB) is a standard tool to analyze the performance of localization algorithm, which provides a lower bound for the position error of any unbiased estimator \cite{ref4}. Conventionally, CRLB assumes a fixed scenario, where the nodes are placed at a particular geometry, and this assumption limits the applicability of CRLB as it cannot properly reflect the impact of the random geometry. To evaluate the localization error of a random network, we use stochastic geometry \cite{ref5, ref5-b} and consider the ensemble average of the node spatial locations. Then the CRLB is no longer a fixed value, but rather a random variable (RV) conditioned on the number of participating BSs, where the randomness of CRLB is induced by the randomness of the nodes. Based on the $L$-localizability and random CRLB, we provide a deep insight for the network operator on how to deploy the BSs to achieve a given localization requirement. The main contributions of this paper are summarized as follows. 

\subsubsection{$L$-Localizability}
We derive the tractable expression of $L$-localizability to study the number of BSs who can participate in a localization procedure. In \cite{ref6}, the authors studied on how the network parameters affect the localization performance of the Long Term Evolution (LTE) cellular network. In this work, we derive the $L$-localizability for the mmWave networks, where the impacts of the directional antenna and Nakagami fading on mmWave-based localization systems are assessed. Furthermore, we introduce asymptotic bounds and approximations for the distribution of the $L$-localizability and CRLB to provide analytical tools to track the performance of localization systems. 

\subsubsection{Random AOA-based CRLB}
In the mmWave networks, accurate AOA measurements can be obtained by using antenna arrays to locate the target of interest with high precision. In this paper, we derive random CRLB for AOA-based positioning. Previous works \cite{ref7, ref8, ref9} applied stochastic geometry to TOA-based localization, and to the best of our knowledge, there is no prior work that investigates random geometry on AOA-based localization systems. We derive the distribution of AOA-positioning based CRLB for the mmWave networks by using stochastic geometry and order statistics. The obtained distribution shows how the network parameters affect localization performance in the mmWave wireless networks.

The rest of this paper is organized as follows. Relevant works are reviewed in Section II, the system model is presented in Section III, and we analyze the localization performance in Section IV. Numerical results are provided in Section V and we conclude the paper in Section VI.

\section{Related Work}
The major localization techniques in the LTE mobile network are TDOA \cite{ref10, ref10-b}, uplink TDOA \cite{ref11}, measurement report (MR) \cite{ref12} and enhanced cell ID (E-CID) \cite{ref13}. Compared with the LTE mobile network, mmWave is regarded as a promising candidate to meet demands for achieving accurate localization in the 5G mobile network. Conventionally, the localization approaches can be divided into two categories: direct and indirect localization. In the mmWave networks, we focus on the latter due to the high computational complexity of the former. The target of interest can be located in mmWave networks using the indirect approach by estimating the channel parameters, including TOA, AOA, and RSS \cite{ref14}. Based on the processing methods of mmWave signals, localization approaches can be categorized into proximity, fingerprinting, and geometry-based \cite{ref15}. In this paper, we mainly focus on the geometry-based positioning approach because large-scale antenna arrays can provide high angular resolution \cite{ref16}. 

Localization performance is generally evaluated using CRLB for a fixed geometry \cite{ref4}. For considering all possible localization scenarios, we aim to derive the network-wide distribution of localization performance, and there are two important metrics, namely, the probability that a given number of BSs can participate in a localization procedure, and the distribution of the CRLB conditioned on the number of participating BSs. The first metric, which includes finding the participation probability of a given number of BSs, was studied in \cite{ref6}. The authors modeled a cellular network with a homogeneous Poisson point process (PPP) \cite{ref5} and applied a “dominant interferer analysis” to derive an expression for the probability of $L$-localizability. However, this method is only suitable for LTE mobile networks. Compared with \cite{ref6}, we derive an accurate expression of $L$-localizability using Alzer’s inequality \cite{ref17} for characterizing the localization performance of mmWave wireless networks.

Regarding the second metric, there have been several attempts in the literature to achieve this conditional distribution of CRLB. In \cite{ref18}, approximations of this conditional distribution were presented for RSS and TOA localization systems. However, these distributions are sensitive to the number of participating BSs, and it is only accurate for numerous participating BSs. In real-world scenarios, we prefer to measure the conditional distribution using smaller number of participating BSs because this is more common in cellular networks. Additionally, \cite{ref7} presented an analysis of how the CRLB is affected by the order statistics of internodal angles. This analysis reveals a connection between the second largest internodal angle and the CRLB, leading to an accurate approximation of the CRLB. However, only TOA-based localization is considered in a general fading channel which takes the large-scale fading into consideration. Motivated by these works, we explore the localization performance using the AOA measurements and apply it in the mmWave-based cellular network. Different from \cite{ref7}, we analyze how the CRLB is affected by the ordered distances between BSs and target, and an accurate approximation of the CRLB is provided using the $\lceil L/4 \rceil$-th distance between these ordered distances, where $L$ is the number of participating BSs in a localization procedure. 

\section{System Model}

In this section, we describe the system model where the key notations used in this paper are summarized in Table $1$.

\subsection{Network Model}
We consider downlink transmission in a mmWave cellular network where the locations of BSs are modeled using a homogeneous PPP\cite{ref5}. As illustrated in Fig. 1, we assume that the target is located at the origin \textit{O} and the BSs are randomly distributed over the $\mathbb{R}^{2}$ plane. The red triangle represents the nearest BS to the target that is located inside the disk, whereas the green triangle indicates the furthest BS from the target residing in the disk. Furthermore, blue and yellow triangles represent the BSs that are located inside and outside the disk, respectively. Let us denote the locations of the BS as $\bm{\psi}_{l} = [x_{l}, y_{l}] \in \mathbb{R}^{2}$ and the distance between the $l$-th BS and target as $r_{l} = ||\bm{\psi}_{l}||$. Based on the system model, the probability density function (PDF) and cumulative distribution function (CDF) of the $L$-th nearest BS are given by \cite{ref19}
\begin{equation}
\begin{split}
    f_{r_{L}}(r) &= \frac{2(\lambda\pi r^{2})^{L}}{r(L-1)!}e^{-\lambda\pi r^{2}},\\
    F_{r_{L}}(r) &= 1 - \sum_{n=0}^{L-1}\frac{1}{n!}e^{-2\pi\lambda r^{2}}(2\pi\lambda r^{2})^{n},
\end{split}
\label{eq-1}
\end{equation}
where $\lambda$ represents the BS density. Conditioned on the distance of the $L$-th BS from \textit{O}, the remaining BSs closer to the origin than the $L$-th BS form a uniform binomial point process (BPP) on $\bm{b}(\textit{O},R_{L})$ \cite{ref20}, where the PDF and CDF of $r_{l}$ are given by 
\begin{equation}
    f_{r_{l}}(r) = \frac{2r}{r_{L}^{2}-r_{1}^{2}}, \quad F_{r_{l}}(r) = \frac{r^{2}}{r_{L}^{2}-r_{1}^{2}},
\label{eq-2}
\end{equation}
with $r_{1} \leq r_{l} \leq r_{L}$. In Section \uppercase\expandafter{\romannumeral4}, we applied order statistic to obtain the distribution of the ordered distances.
\begin{table}[t]
\caption{Summary of Notation}
\label{table}
\setlength{\tabcolsep}{3pt}
\begin{tabular}{|p{45pt}|p{175pt}|}
\hline
Notation & Meaning \\
\hline
${}^T$ & transpose \\
${}^H$ & conjugate transpose \\
$||\cdot||$ &  Euclidean norm \\
$\bm{\psi}_{l}$ & location of $l$-th BS \\
$\bm{\psi}_{t}$ & location of target \\
$r_{l}$ & distance between $l$-th BS and target \\
$r_{1}$ & distance between closest BS and target \\
$r_{L}$ & distance between furthest BS and target \\
$\lambda$ & BS density in the disk \\
$P_{t}$ & BS and target transmit power\\
$N_{t}$ & number of antenna elements \\
$N$ & number of clusters \\
$\rho_{\bm{\psi}_{n}}$ & small-scale fading gain \\
$d$ & antenna spacing \\
$\lambda_{w}$ & antenna wavelength \\
$\theta_{\bm{\psi}}$ & AOA of BS at location $\bm{\psi}$ \\
$n_{{\rm{AOA}},l}$ & WGN with zero mean and variance of $\sigma_{{\rm{AOA}},l}^{2}$ \\
$G_{1}$ & main-lobe gain \\
$G_{2}$ & side-lobe gain \\
$p_{a}$ & probability of main-lobe gain is received \\
$p_{b}$ & probability of side-lobe gain is received \\
$P_{T}$ & total transmit power \\
$P_{m}$ & power spectrum density of main-lobe \\
$P_{s}$ & power spectrum density of side-lobe \\
$\sigma_{n}^{2}$ & normalized noise power\\
$a_{i}$ & network load indicator \\
$q$ & probability of BS is activated \\
$\alpha$ & path-loss exponent\\
$\Omega$ & total number of activated BSs \\
$\tau$ & signal-to-interference-plus-noise ratio threshold \\
$\gamma$ & maximum number of selectable BSs \\
$\sigma_{\rm{AOA}}$ & standard deviation of AOA measurement \\
$G_{c}$ & average channel gain of mmWave network\\
$N_{0}$ & spectral density of WGN \\
$W_{\rm{TOT}}$ & total mmWave system bandwidth \\
\hline
\end{tabular}
\label{tab1}
\end{table}

\subsection{Channel Model}
We assume that each BS is equipped with a directional antenna array composed of $N_{t}$ elements and all BSs operate at a constant power $P_{t}$. In the mmWave channel, the non-line-of-sight (NLOS) interference is negligible since the channel gains of NLOS paths are typically 20 dB weaker than those from the line-of-sight (LOS) \cite{ref21}. The effect of path-loss can be reduced due to the utilization of the directional antenna arrays, and it is also applied to provide highly directional beams. The received signal from the $l$-th BS to the origin is given by
\begin{equation}
\begin{split}
y(t) &= \sqrt{P_{t}\beta}\bm{h}_{\psi_{l}}\bm{w}_{\psi_{l}}r_{l}^{-\frac{\alpha}{2}}s_{\psi_{l}}(t) + n(t)\\
&+ \sum_{\bm{\psi} \in \bm{\psi}^{\prime}}\sqrt{P_{t}\beta}\bm{h}_{\bm{\psi}}\bm{w}_{\bm{\psi}}||\bm{\psi}||^{-\frac{\alpha}{2}}s_{\bm{\psi}}(t), \quad t \in\left[0, T \right],
\end{split}
\label{eq-3}
\end{equation}
where $s_{\bm{\psi}}(t)$ is the transmit signal, $\bm{h}_{\bm{\psi}}$ is the channel vector, $\alpha$ and $\beta$ respectively represent the path-loss exponent and path-loss intercept, $\bm{w}_{\bm{\psi}}$ denotes the beamforming vector of the node at location $\bm{\psi}$, and $n(t)$ represents the additive white Gaussian noise (AWGN) with variance $\sigma^{2}$. Note that the locations of the interfering transmitters are denoted as $\bm{\psi}^{\prime}$.

Due to high free-space path-loss, the mmWave propagation environment is well characterized by a clustered channel model, known as the Saleh-Valenzuela model\cite{ref22}:
\begin{equation}
\bm{h}_{\bm{\psi}} = \sqrt{N_{t}}\sum_{n=1}^{N}\rho_{\bm{\psi},{n}}\bm{a}_{t}^{H}(\theta_{\bm{\psi},{n}}),
\label{eq-4}
\end{equation}
where $N$ is the number of clusters and $\rho_{\bm{\psi},{n}}$ represents the complex small-scale fading coefficient of the $n$-th cluster. We assume that the fading channel power gain follows a gamma distribution, \textit{i.e.}, $|\rho_{\bm{\psi}}|^{2} \sim \Gamma(M,\frac{1}{M})$, with Nakagami parameter $M$. In this paper, we focus on LOS paths, \textit{i.e.}, $N = 1$, and adopt a uniformly random single path (UR-SP) channel model that is commonly used in the mmWave network analysis \cite{ref23}. The $\bm{a}_{t}(\theta_{\bm{\psi}})$ represents the transmit array response vector corresponding to the AOA $\theta_{\bm{\psi}}$. We consider a uniform linear array (ULA) with $N_{t}$ antenna elements, where the transimt array response vectors are given by
\begin{equation}
\bm{a}_{t}(\theta_{\bm{\bm{\psi}}}) = \frac{1}{\sqrt{N_{t}}}\left[1, \ldots, e^{j2\pi k\theta_{\bm{\psi}_{1}}}, \ldots ,e^{j2\pi (N_{t}-1)\theta_{\bm{\psi}}}\right]^{T},
\label{eq-5}
\end{equation}
where $d$ is the antenna spacing, $\lambda_{w}$ represents the wavelength, $\phi_{\bm{\psi}}$ denotes the AOA, $k \in \left[0, N_t \right]$ is the antenna index, and $\theta_{\bm{\psi}} = \frac{d}{\lambda_{w}}\sin\phi_{\bm{\psi}}$ is uniformly distributed over $\left[-\frac{d}{\lambda_{w}},\frac{d}{\lambda_{w}}\right]$. 

\begin{figure}[t]
\centerline{\includegraphics[width=0.85\columnwidth]{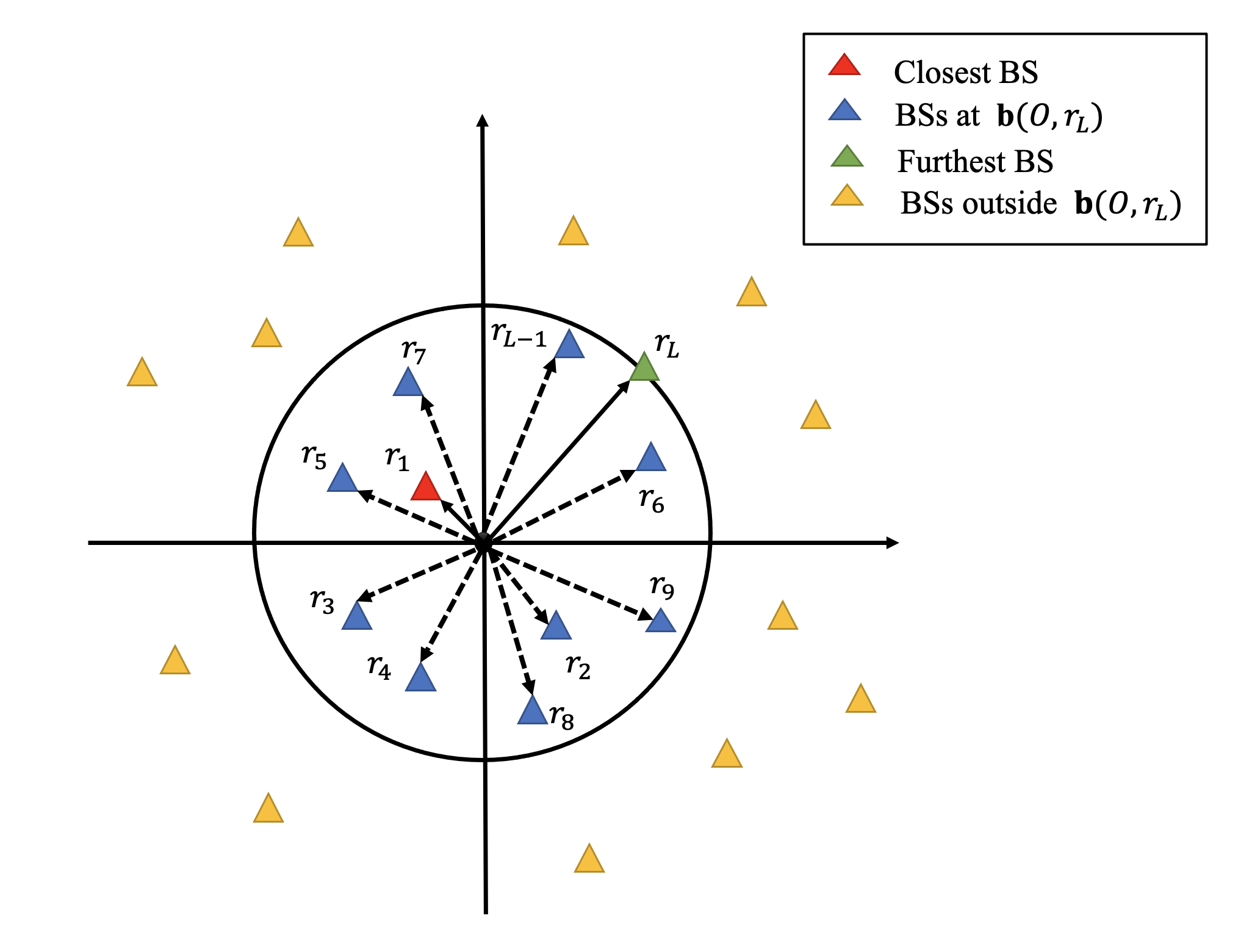}}
\caption{System model of the mmWave wireless networks}
\label{fig1}
\end{figure}

\begin{figure}[t]
\centerline{\includegraphics[width=0.85\columnwidth]{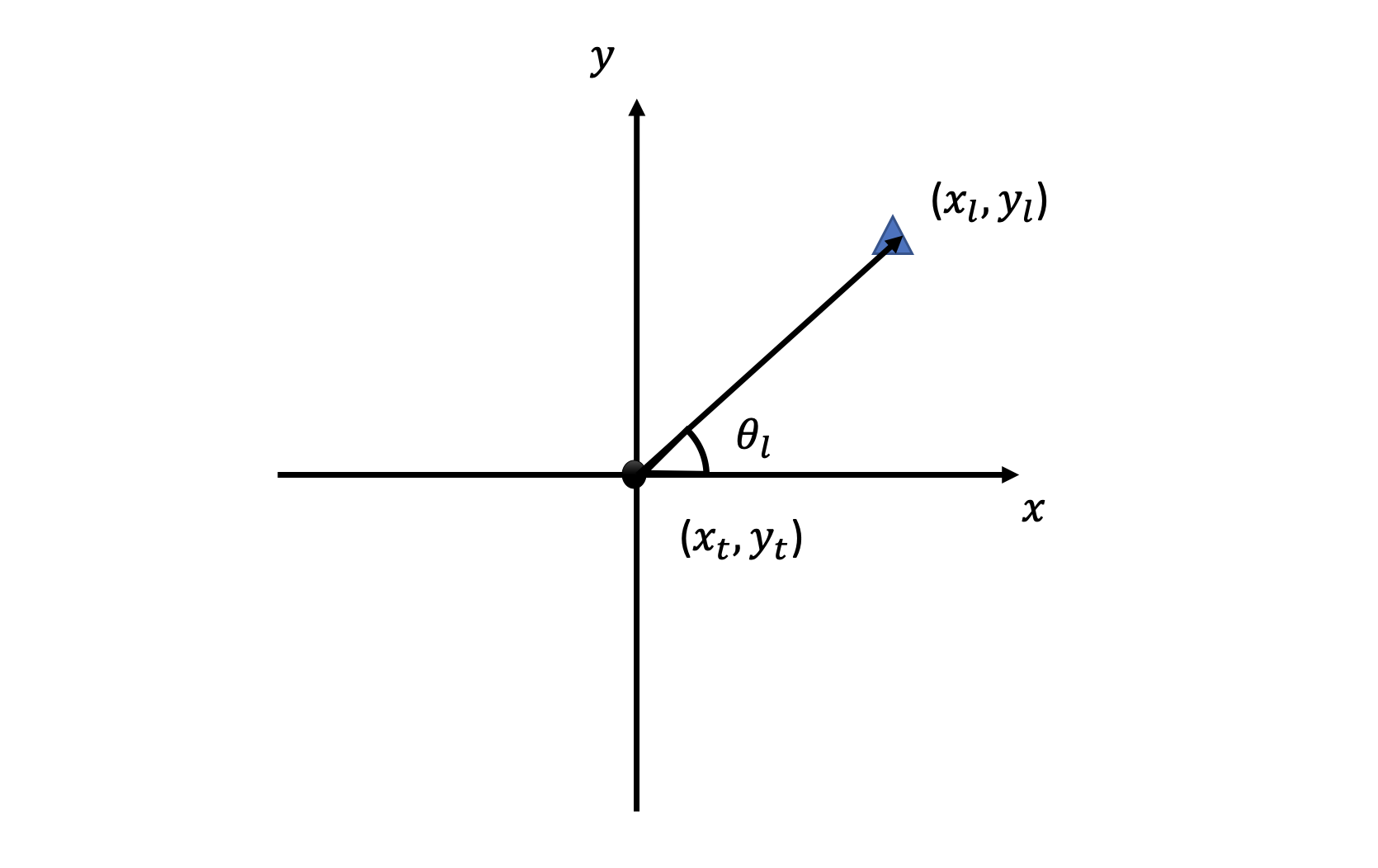}}
\caption{AOA-based positioning}
\label{fig2}
\end{figure}

Once the AOA measurements are obtained, we compute the location of the target, where we assume LOS propagation. As shown in Fig. 2, we denote the AOA between the target and $l$-th BS as $\theta_{l}$ and the location of target as $\bm{\psi}_{t} = [x_{t},y_{t}]$, 
\begin{equation}
\tan(\theta_{l})=\frac{y_{l}-y_{t}}{x_{l}-x_{t}},\quad l=\{1, \ldots, L\}.
\label{eq-6}
\end{equation}
The AOA measurement at the $l$-th BS is modeled as follows
\begin{equation}
r_{{\rm{AOA}},l} = \theta_{l} + n_{{\rm{AOA}},l} = \tan^{-1}\left(\frac{y_{l}-y_{t}}{x_{l}-x_{t}}\right) + n_{{\rm{AOA}},l},
\label{eq::eq7}
\end{equation}
where $n_{{\rm{AOA}},l}$ is the AWGN with variance $\sigma_{{\rm{AOA}},l}^{2}$. The AOA measurements in (\ref{eq::eq7}) can be represented by a vector form
\begin{equation}
{\bm{r}}_{\rm{AOA}} = \bm{f}_{AOA}(\bm{\psi}) + \bm{n}_{{\rm{AOA}},l},
\label{eq}
\end{equation}
where ${\bm{r}}_{\rm{AOA}}$, ${\bm{n}}_{\rm{AOA}}$, and $\bm{f}_{AOA}(\bm{\psi})$ are respectively defined by
\begin{equation}
    \begin{split}
        {\bm{r}}_{\rm{AOA}} &= \left[r_{{\rm{AOA}},1}, r_{{\rm{AOA}},2}, \ldots, r_{{\rm{AOA}},L}\right]^T,\\
        {\bm{n}}_{\rm{AOA}} &= \left[n_{{\rm{AOA}},1}, n_{{\rm{AOA}},2}, \ldots, n_{{\rm{AOA}},L}\right]^T,\\
        \bm{f}_{AOA}(\bm{\psi}) &= \left[ \tan^{-1}\left(\frac{y_{1}-y_{t}}{x_{1}-x_{t}}\right), \ldots, 
                                  \tan^{-1}\left(\frac{y_{L}-y_{t}}{x_{L}-x_{t}}\right) \right]^T.    
  \end{split}
    \label{eq::eq8}
\end{equation}


\subsection{Analog Beamforming and Antenna Radiation Pattern}
Assuming that the AOA of the channel between the BS at location $\bm{\psi}_{l}$ and its serving user at location $\bm{\psi}_{t}$ is $\theta_{\bm{\psi}_{l}}$, the beamforming vector is given by
\begin{equation}
\bm{w}_{\bm{\psi}_{l}} = \bm{a}_{t}(\theta_{\bm{\psi}_{l}}),
\label{eq}
\end{equation}
which means that the BS should align the beam direction exactly with the propagation channel to obtain the maximum power gain. However, the beam direction cannot always align with
the transmit signal. Hence, we consider the single main-lobe and single side-lobe at antennas of both BS and mobile user, and all lobes are approximated by a flat-top antenna pattern\cite{ref24}. That is, the single main-lobe with beam-width $\theta_{1}$ has antenna gain $G_{1}$ and each side lobe with identical beam-width $\theta_{2}$ has antenna gain $G_{2}$. We assume that the power spectrum density (PSD) of the main-lobe and side-lobe at a distance $r$ are denoted as $P_{m}$ and $P_{s}$. Hence, the total transmit power $P_{T}$ consists of the 
main-lobe and side-lobe radiation powers, which is given by \cite{ref24}
\begin{equation}
P_{T} = P_{m}2\pi r^{2}\left[1-\cos\frac{\theta_{1}}{2}\right]+N_{t}P_{s}2\pi r^{2}\left[1-\cos\frac{\theta_{2}}{2}\right],
\label{eq}
\end{equation}
where $P_{m} = G_{1}P_{T}/4\pi r^{2}$ and $P_{s} = G_{2}P_{T}/4\pi r^{2}$. Let us denote $k = G_{2}/G_{1}$ where $k \in (0,1)$, \textit{i.e.}, $G_{2} = kG_{1}$. 

For the associated signal transmission, we assume perfect alignment where both the BS and user utilize the main-lobe, achieving the squared gain $G_{1}^{2}$. For the interfering signal, the interfering BSs are randomly distributed in $[0, 2\pi)$. The transmit antenna gain $G_{Tx}$ at the transmitter and the receive antenna gain $G_{Rx}$ at the receiver are randomly chosen from a discrete set $\left\{G_{1}, G_{2} \right\}$ with probability $p_{a} = \frac{\theta_{1}}{2\pi}$ and $p_{b} = 1-p_{a}$, respectively. Let $G_{TRx} = G_{Tx}G_{Rx}$, we have
\begin{equation}\label{eq}
G_{TRx} = 
\left\{
\begin{array}{ll}
G_{1}^{2},  \quad &p_{1} = p_{a}^{2} \\
G_{1}G_{2}, \quad &p_{2} = 2p_{a}p_b \\
G_{2}^{2},  \quad &p_{3} = p_b^2. 
\end{array}
\right.
\end{equation}
Based on the antenna radiation pattern, the product of small-scale fading gain and beamforming gain of the BS at location $\bm{\psi}$ is computed as:
\begin{equation}
|\bm{h}_{\bm{\psi}}\bm{w}_{\bm{\psi}}|^{2} = N_{t}|\rho_{\bm{\psi}}|^{2}G_{TRx}.
\label{eq}
\end{equation}

\section{Performance Analysis}
In this section, we analyze the performance of AOA-based localization over a mmWave network. To evaluate the localization performance, we will introduce two metrics; $L$-localizability and AOA-based random CRLB.

\subsection{L-Localizability}
A target is \textit{localizable} if there are a sufficient number of participating BSs such that the localization procedure can be conducted. We introduce \textit{L-Localizability}, which is a probability to have $L$ localizable BSs within the network \cite{ref6}. Conventionally, commonly-accepted minimum values of $L$ for the unambiguous operation of a localization system are $2$, $3$, $3$ and $3$ for AOA, TOA or RSS, and TDOA, respectively \cite{ref4}. If we treat the interference originated from outside of the  circular disk with radius $R_{L}$ as a noise, the signal-to-interference-plus-noise ratio (SINR) of the link from the $k$-th BSs to the target can be expressed as a function of $L$ as
\begin{equation}
{\rm{SINR}}_{k}(L) = \frac{G_{1}^{2}|\rho_{\bm{\psi}_{k}}|^2r_{k}^{-\alpha}}{\sigma_{n}^2+J},
\label{eq}
\end{equation}
where $\sigma_{n}^{2} = \frac{\sigma_{T}^{2}+\sigma_{{\rm{out}}}^{2}}{\beta P_{t}N_{t}}$ is the normalized noise power, including the thermal noise power $\sigma_{T}^{2}$ and interference power $\sigma_{{\rm{out}}}^{2}$ outside the circular disk. The interference from nodes inside the disk, denoted by $J$, is expressed as:
\begin{equation}
J = \sum_{i=1, i \neq k}^{L-1}a_{i}G_{TRx,i}|\rho_{\bm{\psi}_{i}}|^2||\bm{\psi}_{i}||^{-\alpha}, 
\label{eq}
\end{equation}
where $a_i \in \{0, 1\}$ is utilized to simulate the network load. The probability of $a_i=1$ equals $q$ which is the probability of a BS inside the disk to be activated. The $a_{i}$ represents whether the BS is activated in the localization procedure and we assume that the activation probability $P\left(a_i = 1\right) = q$ is fixed throughout the localization procedure. 

For a given ${\bm{\psi}}\in\mathbb{R}^{2}$, a mobile device is said to be $L$-localizable if at least $L$ BSs participate in the localization procedure. Let us denote the SINR threshold as $\tau$ and the maximum number of BSs that can participate in the localization procedure as $\gamma$, defined as 
\begin{equation}
\gamma = {\mathop{\arg\max}_{L}}\left(L\cdot\prod_{k=1}^{L}\mathbb{I}\left({\rm{SINR}}_{k}(l)\geq\tau\right)\right),
\label{eq}
\end{equation}
where $\mathbb{I}(.)$ is the indicator function. Then, the $L$-localizability, denoted by $P_{L}$, is derived as:
\begin{equation}
P_{L} = P\left(\gamma\geq L\right) = \mathbb{E}\left[\prod_{k=1}^{L}\mathbb{I}\left({\rm{SINR}}_{k}(l)\geq\tau\right)\right].
\label{eq-15}
\end{equation}
Since the SINR from a BS farther from the mobile device is lower than that of the closer BS, the following inequality holds: $\mathbb{I}({\rm{SINR}}_{k}(L)\geq\tau)\geq\mathbb{I}({\rm{SINR}}_{l}(L)\geq\tau)$ for all $k\geq l\geq L$. 
Then, the $L$-localizability $P_L$ can be computed as follows
\begin{equation}
\begin{split}
P_{L} &= \mathbb{E}\left[\mathbb{I}({\rm{SINR}}_{L}(L)\geq \tau\right] = P\left({\rm{SINR}}_{L}(L)\geq \tau\right)\\
&= 1-P\left(|\rho_{\psi_{L}}|^2 \leq \frac{\tau ~ r_{L}^{\alpha}}{G_{1}^{2}} \left(\sigma_{n}^2+J\right)\right)\\
&\overset{a}{\simeq}1-\mathbb{E}_{r_{L}}\left[\left(1-e^{-\nu\frac{\tau}{G_{1}^{2}}r_{L}^{\alpha}\left(\sigma_{n}^{2}+J\right)}\right)^{M}\right] \\
&\overset{b}{=} \mathbb{E}_{r_{L}}\left[\sum_{i=1}^{M}(-1)^{i+1} \binom{M}{i}
e^{-s\sigma_n^{2}}\mathcal{L}_{I}(s)\right]\\
&= \int_{0}^{\infty}f_{r_{L}}(r)\sum_{i=1}^{M}(-1)^{i+1}\binom{M}{i} e^{-s\sigma_n^{2}}~\mathcal{L}_{I}(s)dr \\
\end{split}   
\label{eq-16} 
\end{equation}
where $\nu = M(M!)^{-\frac{1}{M}}$, $s = i\nu\frac{\tau}{G_{1}^{2}}r_{L}^{\alpha}$ and $\mathcal{L}_{I}(s) = \mathbb{E}_{I}[e^{-sJ}]$ is the Laplace transform of the interference. Step $(a)$ follows by the Alzer's inequality and $(b)$ is obtained based on the binomial expansion. The Laplace transform of the interference is \cite{ref25}
\begin{equation}
\begin{split}
\mathcal{L}_{I}(s) &=\mathbb{E}_{I}[e^{-sJ}]\\
&= \exp\bigg[-2\pi\lambda q \underbrace{\int_{r_{1}}^{r_{L}}(1-\mathbb{E}_{g_{{\bm{\psi}}}}[e^{-sJ}]){r}dr}_{\triangleq\Lambda}\bigg], \\
\end{split}
\label{eq-17}  
\end{equation}
where $g_{{\bm{\psi}}} = G_{TRx}|\rho_{{\bm{\psi}}}|^{2}$ represents the combined effect of antenna gain and channel gain at the location ${\bm{\psi}}$. The term $\Lambda$ is computed as:
\begin{equation}
\begin{split}
\Lambda &= \int_{r_{1}}^{r_{L}}\left(1-\mathbb{E}_{g_{{\bm{\psi}}}}[e^{-sJ}]\right)r dr \\
&= -\frac{1}{2} \Big[ r_{L}^{2} - \delta r_{L}^{2}\mathbb{E}_{g_{{\bm{\psi}}}}\left[E_{1+\delta}\left(sg_{{\bm{\psi}}}r_{L}^{-\alpha}\right)\right]\\
&\quad - r_{1}^{2} 
+  \delta r_{1}^{2}\mathbb{E}_{g_{{\bm{\psi}}}}\left[E_{1+\delta}\left(sg_{{\bm{\psi}}}r_{1}^{-\alpha}\right)\right]\Big],
\end{split}
\label{eq-18}  
\end{equation}
where $\delta = \frac{2}{\alpha}$ and $E_{1+\delta}(.)$ is the generalized exponential integral \cite{ref26}. The term $\mathbb{E}_{g_{{\bm{\psi}}}}[E_{1+\delta}(s g_{{\bm{\psi}}}r^{-\alpha})]$ is given by 
\begin{equation}
\begin{split}
&\mathbb{E}_{g_{{\bm{\psi}}}}[E_{1+\delta}(sg_{{\bm{\psi}}}r^{-\alpha})] \\
=&\frac{s^{\delta} \Gamma\left(-\delta\right)}{r^{2}}
\left[
\mathbb{E}_{g_{\psi}}\left[g_{\psi}^{\delta}\right] +\frac{\alpha}{2}-\sum_{p=1}^{\infty}\frac{(-s)^{p} \cdot \mathbb{E}_{g_{\psi}}\big[g_{\psi}^{\delta}\big] }{r^{\alpha p}\cdot p! \cdot(p-\delta)}
\right],
\end{split}
\label{eq-19}  
\end{equation}
and the fractional moment of $g_{{\bm{\psi}}}$ is derived as:
\begin{equation}
\begin{split}
&\mathbb{E}_{g_{{{\psi}}}}[g_{{{\psi}}}^{\delta}] =  \mathbb{E}_{|\rho_{{\bm{\psi}}}|^{2},G_{TRx}}\left[\left(|\rho_{{\bm{\psi}}}|^{2}G_{TRx}\right)^{\delta}\right] \\
&= \frac{\Gamma(M+\delta)}{\Gamma(M)M^{\delta}} \cdot \mathbb{E}_{G_{TRx}}(G_{TRx}^{\delta}) \\
&= \frac{\Gamma(M+\delta)}{\Gamma(M)M^{\delta}} \cdot \left[G_{1}^{2\delta}p_{a}^{2}+2(G_{1}G_{2})^{\delta}p_{a}p_{b}+G_{2}^{2\delta}p_{b}^{2}\right].
\end{split}
\label{eq-20}
\end{equation}
Hence, the $L$-localizability can be numerically evaluated by substituting (\ref{eq-17})-(\ref{eq-20}) into (\ref{eq-16}).

\subsection{Approximation of Cram$\bf \acute{\rm e}$r-Rao Lower Bound}
We derive the AOA-based random CRLB using the Fisher information matrix (FIM), denoted by $\bm{{\rm{I}}}(\bm{\psi})$ \cite{ref4}
\begin{equation}
\bm{{\rm{I}}}(\bm{\psi}) = \left(\frac{\partial {\bm{f}}_{\rm{AOA}}(\bm{\psi})}{\partial \bm{\psi}}\right)^{T}\bm{C}_{\rm{AOA}}^{-1}
\frac{\partial {\bm{f}}_{\rm{AOA}}(\bm{\psi})}{\partial \bm{\psi}}, 
\label{eq}
\end{equation}
where $\bm{C}_{\rm{AOA}}^{-1}$ represents the inverse of the noise covariance matrix and the derivative of 
${\bm{f}}_{\rm{AOA}}(\bm{\psi})$ which is the angle vector with respect to $\bm{\psi}$ are given by
\begin{equation}
\begin{split}
\bm{C}_{\rm{AOA}}^{-1} &= \text{diag}\left(\frac{1}{\sigma_{{\rm{AOA}},1}^{2}},\frac{1}{\sigma_{{\rm{AOA}},2}^{2}},\cdots,\frac{1}{\sigma_{{\rm{AOA}},L}^{2}}\right),\\
\newline \\
\frac{\partial {\bm{f}}_{\rm{AOA}}(\bm{\psi})}{\partial \bm{\psi}} 
&= -\begin{bmatrix}
\frac{y-y_{1}}{(x-x_{1})^{2}+(y-y_{1})^{2}} & \frac{x-x_{1}}{(x-x_{1})^{2}+(y-y_{1})^{2}} \\
\frac{y-y_{2}}{(x-x_{2})^{2}+(y-y_{2})^{2}} & \frac{x-x_{2}}{(x-x_{2})^{2}+(y-y_{2})^{2}} \\
\vdots & \vdots \\
\frac{y-y_{L}}{(x-x_{L})^{2}+(y-y_{L})^{2}} & \frac{x-x_{L}}{(x-x_{L})^{2}+(y-y_{L})^{2}} 
\end{bmatrix}.
\end{split}
\label{eq}
\end{equation}

Without loss of generality, $\sigma_{\rm{AOA}}$ is considered to be a known quantity and assumed to be identical for each BSs, \textit{i.e.}, $\sigma_{{\rm{AOA}},1} = \sigma_{{\rm{AOA}},2} = \cdots = \sigma_{{\rm{AOA}},L}$\cite{ref7}. The numerical value of $\sigma_{\rm{AOA}}$ depends on the average SNR of the mmWave networks, denoted by $\rm{\overline{SNR}}$, as follows
\begin{equation}
{\rm{\overline{SNR}}} = \frac{G_{c}P_{t}}{N_{0}W_{\rm{TOT}}}
\label{eq}
\end{equation}
where $G_{c}$ is the average channel gain, $N_{0}$ is the spectral density of the WGN, and $W_{\rm{TOT}}$ is the total system bandwidth \cite{ref27, ref28}. Hence, $\bm{{\rm{I}}}(\bm{\psi})$ is  
\begin{equation}
\begin{split}
&\bm{{\rm{I}}}_{\rm{AOA}}(\bm{\psi})  \\
= &\sigma_{\rm{AOA}}^{2}
\begin{bmatrix}
\sum_{i=1}^{L}\frac{(y-y_{i})^{2}}{r_{i}^{4}} & -\sum_{i=1}^{L}\frac{(x-x_{i})(y-y_{i})}{r_{i}^{4}} \\
-\sum_{i=1}^{L}\frac{(x-x_{i})(y-y_{i})}{r_{i}^{4}} & \sum_{i=1}^{L}\frac{(x-x_{i})^{2}}{r_{i}^{4}} 
\end{bmatrix}.
\end{split}
\label{eq-24}
\end{equation}

To assess the distribution of the CRLB, we introduce the position error bound (PEB), which is the square root of the CRLB \cite{ref29}. We will denote the PEB by $S$ and its closed-form expression can be obtained by using (\ref{eq-24})
\begin{equation}
    \begin{split}
        S &\triangleq \sqrt{\rm{CRLB}} = \sqrt{{\rm{tr}}(\bm{I}_{\rm{AOA}}^{-1}(\bm{\psi}))} = \sigma_{\rm{AOA}}\frac{\sqrt{L}}{\sqrt{Q_{1}-Q_{2}}},
    \end{split}
\label{eq-25}
\end{equation}
where $Q_{1}$ and $Q_{2}$ are
\begin{equation}
\begin{split}
Q_{1} &=  \sum_{i=1}^{L}\frac{(y_{i}-y_{t})^{2}}{r_{i}^{4}}\sum_{j=1}^{L}\frac{(x_{j}-x_{t})^{2}}{r_{j}^{4}}, \\
Q_{2} &= \sum_{i=1}^{L}\frac{(x_{i}-x_{t})^{2}(y_{i}-y_{t})^{2}}{r_{i}^{8}}.
\end{split}
\label{eq-26}
\end{equation}
Since (\ref{eq-25}) and (\ref{eq-26}) are functions of $2L$ random variables, \textit{i.e.}, $\left(x_i, y_i\right)$ for $1 \leq i \leq L$, we need to simplify (\ref{eq-26}) using its asymptotic bounds, which will enable us to characterize the distribution of (\ref{eq-25}). In the following proposition, we derived a tight upper bound for $Q_{1}-Q_{2}$ and through  simulation, we verified that approximation error is less than $5\%$ for $L \geq 8$. 

\begin{prop}
The random variable $Q_{1}-Q_{2}$ from (\ref{eq-25}) can be upper bounded as follows
\begin{equation}
Q_{1}-Q_{2} \leq \frac{1}{4}\left[\left(\sum_{i=1}^{L}\frac{1}{r_{i}^{2}}\right)^2 - \sum_{i=1}^{L}\frac{1}{r_{i}^{4}}\right]
\label{eq-27}
\end{equation}
\end{prop}

\begin{proof}
First, we derive the lower bound of $Q_2$ as follows
\begin{equation}
    \begin{split}
        Q_{2} 
        &\overset{(a)}{\geq} \sum_{i=1}^{L} \frac{\frac{1}{4}[(x_{i}-x_{t})^{2}+(y_{i}-y_{t})^{2}]^{2}}{r_{i}^{8}}
        \overset{(b)}{=} \sum_{i=1}^{L}\frac{1}{4r_{i}^{4}},
    \end{split}
    \label{eq-27b}
\end{equation}
where the inequality $(x_{i}-x_{t})^{2}+(y_{i}-y_{t})^{2} \geq 2(x_{i}-x_{t})(y_{i}-y_{t})$ is applied to step (a) and Cartesian coordinates is converted to polar coordinate in step (b). As shown in Fig. 2, the polar coordinate of $(x_i, y_i)$ is given by
\begin{equation}
    \begin{split}
        x_i - x_t = r_i \cos\left(\theta_i\right), \quad 
        y_i - y_t = r_i \sin\left(\theta_i\right).
    \end{split}
    \label{eq-28}
\end{equation}

Next, we derive the upper bound of $Q_1$
\begin{equation}
    \begin{split}
    Q_1 &= \sum_{i=1}^{L}\frac{(y_{i}-y_{t})^{2}}{r_{i}^{4}}\sum_{j=1}^{L}\frac{(x_{j}-x_{t})^{2}}{r_{j}^{4}}\\
    &= 
\sum_{i=1}^{L}\frac{\sin^2\left(\theta_i\right)}{r_{i}^{2}}
\sum_{j=1}^{L}\frac{\cos^2\left(\theta_j\right)}{r_{j}^{2}},
    \end{split}
    \label{eq-28-b}
\end{equation}
where we will maximize $Q_1$ with respect to the phase $\{\theta_i\}$  for a given distance $\{\theta_i\}$. Then, (\ref{eq-28-b}) can be expressed as 
\begin{equation}
    \begin{split}
    Q_1 &= \sum_{i=1}^{L}\frac{\sin^2\left(\theta_i\right)}{r_{i}^{2}}
\sum_{j=1}^{L}\frac{1 - \sin^2\left(\theta_j\right)}{r_{j}}
= \xi \left(\sum_{j=1}^{L} \frac{1}{r_j^2} - \xi \right),
    \end{split}
    \label{eq-28-c}
\end{equation}
where we denote $\xi \triangleq \sum_{i=1}^{L}\frac{\sin^2\left(\theta_i\right)}{r_{i}^{2}}$. The first order derivative of $Q_1$ is zero when $\xi^{\ast} = \frac{1}{2}\sum_{i=1}^{L}\frac{1}{r_{i}^{2}}$ and the second order derivative of $Q_1$ has a negative value at $\xi^{\ast}$ as follows
\begin{equation}
    \begin{split}
    \frac{\partial Q_1}{\partial \xi} &= \sum_{i=1}^{L}\frac{1}{r_{i}^{2}} - 2 \xi = 0 ~\Rightarrow~
     \xi^{\ast} = \frac{1}{2}\sum_{i=1}^{L}\frac{1}{r_{i}^{2}},
    \\
    \frac{\partial^2 Q_1}{\partial \xi^2} &= - 2 < 0.
    \end{split}
    \label{eq-28-d}
\end{equation}
Hence, the upper bound of $Q_1$ is given by 
\begin{equation}
    \begin{split}
    Q_1 \leq \max_{\{\theta_i\}}Q_1\bigg|_{\xi = \xi^{\ast}} = \left(\frac{1}{2}\sum_{i=1}^{L}\frac{1}{r_{i}^{2}}\right)^2.
    \end{split}
    \label{eq-28-e}
\end{equation}
We obtain (\ref{eq-27}) by (\ref{eq-27b}) and (\ref{eq-28-e}). This completes the proof.
\end{proof}

Based on Proposition 1, the PEB $S$ is lower bounded by
\begin{equation}
    \begin{split}
        S &\geq \frac{2 \sigma_{\rm{AOA}} \cdot \sqrt{L}}{\sqrt{
        \left(\sum_{i=1}^{L}\frac{1}{r_{i}^{2}} \right)^2 - \sum_{i=1}^{L}\frac{1}{r_{i}^{4}}}
        } 
        = \frac{2 \sigma_{\rm{AOA}} \cdot \sqrt{L}}{\sqrt{
        \sum\limits_{\substack{\scriptstyle i, j = 1\\\scriptstyle i\neq j}}^{L}
        \frac{1}{r_{i}^2 r_{j}^2}
        }}.
\end{split}
\label{eq-31}
\end{equation}
In the following assumption, we introduced an approximation of (\ref{eq-31}), which provides a tractable asymptotic bound of $S$. Through simulation, we justified the approximation accuracy. 

\begin{assumption}
Assume that the link distances are sorted in an ascending order, \textit{i.e.}, $R = \left[r_{1}, \cdots, r_{L}\right]$ and $r_{1} \leq r_{2} \leq \cdots \leq r_{L}$. The denominator of (\ref{eq-31}) can be approximated as follows
\begin{equation}
D \triangleq
        \sum\limits_{\substack{\scriptstyle i, j = 1\\\scriptstyle i\neq j}}^{L}
        \frac{1}{r_{i}^2 r_{j}^2}
 \approx \frac{L(L-1)}{r_{\lceil L/4 \rceil}^{4}},
\label{eq-32}
\end{equation}
where $r_{\lceil L/4 \rceil}$ is the $\lceil L/4 \rceil$-th link distance in the ordered set $R = \left[r_{1}, \cdots, r_{L}\right]$ and $L$ is the number of participating BSs.
\end{assumption}

\begin{remark}
We validated (\ref{eq-32}) through simulation, where we repeated the realization of the anchor nodes $10$ million times. Since the set $R$ is sorted, the term $D$ in (\ref{eq-32}) is bounded by 
\begin{equation}
    \begin{split}
    \frac{L\left(L-1\right)}{r_L^4} \leq 
        \sum\limits_{\substack{\scriptstyle i, j = 1\\\scriptstyle i\neq j}}^{L}
        \frac{1}{r_{i}^2 r_{j}^2}
        \leq 
        \frac{L\left(L-1\right)}{r_1^4}. 
    \end{split}
    \label{eq:eq-35}
\end{equation}
We attempt to find the $k$-th term $r_k$ in set $R$ that provides the most accurate approximation to $D$. To solve this problem, we used heuristic approach and evaluated the mutual information between $D$ and the random variable $r_k$ for a given $L$ as follows
\begin{equation}
    \begin{split}
        \min_{1 \leq k \leq L}\mathbb{E}\left[\bigg\lvert D - \frac{L(L-1)}{r_{k}^{4}}\bigg\rvert^2\right] ~\Leftrightarrow~  \max_{1 \leq k \leq L} I(D;r_{k}|L=l), 
    \end{split}
    \label{eq:eq-37}
\end{equation}
where the mutual information conditioned on $L$ is defined as
\begin{equation}
I\left(D;r_{k}|L=l\right) = h\left(D|L=l\right) - h\left(D|r_{k},L=l\right), 
\label{eq-34}
\end{equation}
and the differential entropies are given by
\begin{equation}
\begin{split}
h(D|L=l) &= -\sum\limits_{d\in \rm{D}}f_{D}(d|\mathit{l})\log_{2}{f_{D}(d|\mathit{l})}, \\
h(D|r_{k},L=l) &= -\sum\limits_{\substack{r_{k}\in R_{k},\\d\in\rm{D}}}  f_{D;r_{k}}(d|r, l)\log_{2}{f_{D}(d|r,l)},
\end{split}
\label{eq-35}
\end{equation}
where $R_{k}$ and $\rm{D}$ are the supports of $r_{k}$ and $d$, respectively\cite{ref30}. This approach, which was motivated by \cite{ref31}, can search for the $r_{l}$ that contains the most information of $D$. Through an extensive simulation across a range of $L$, we observed that the $\lceil L/4 \rceil$-th distance maximizes the mutual information as illustrated in Fig. \ref{fig::MI}, which justifies Assumption 1, and thus we can use the $\lceil L/4 \rceil$-th distance to approximate $D$. 
\begin{figure}
\centerline{\includegraphics[width=\columnwidth]{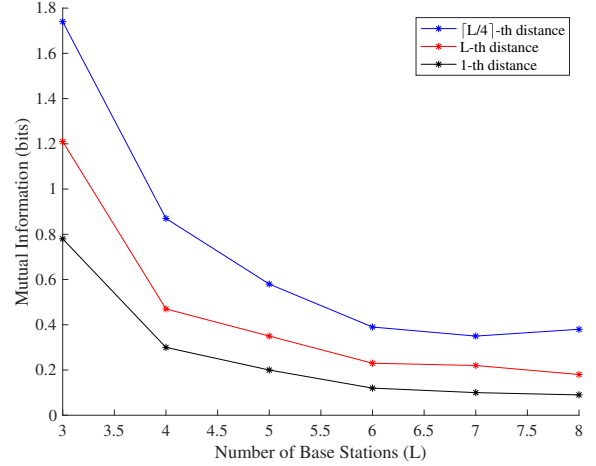}}
\caption{Impact of distance selection on mutual information}
\label{fig::MI}
\end{figure}
\end{remark}

\begin{figure*}
  \centering
  \begin{minipage}[b]{0.48\textwidth}
    \includegraphics[width=\columnwidth]{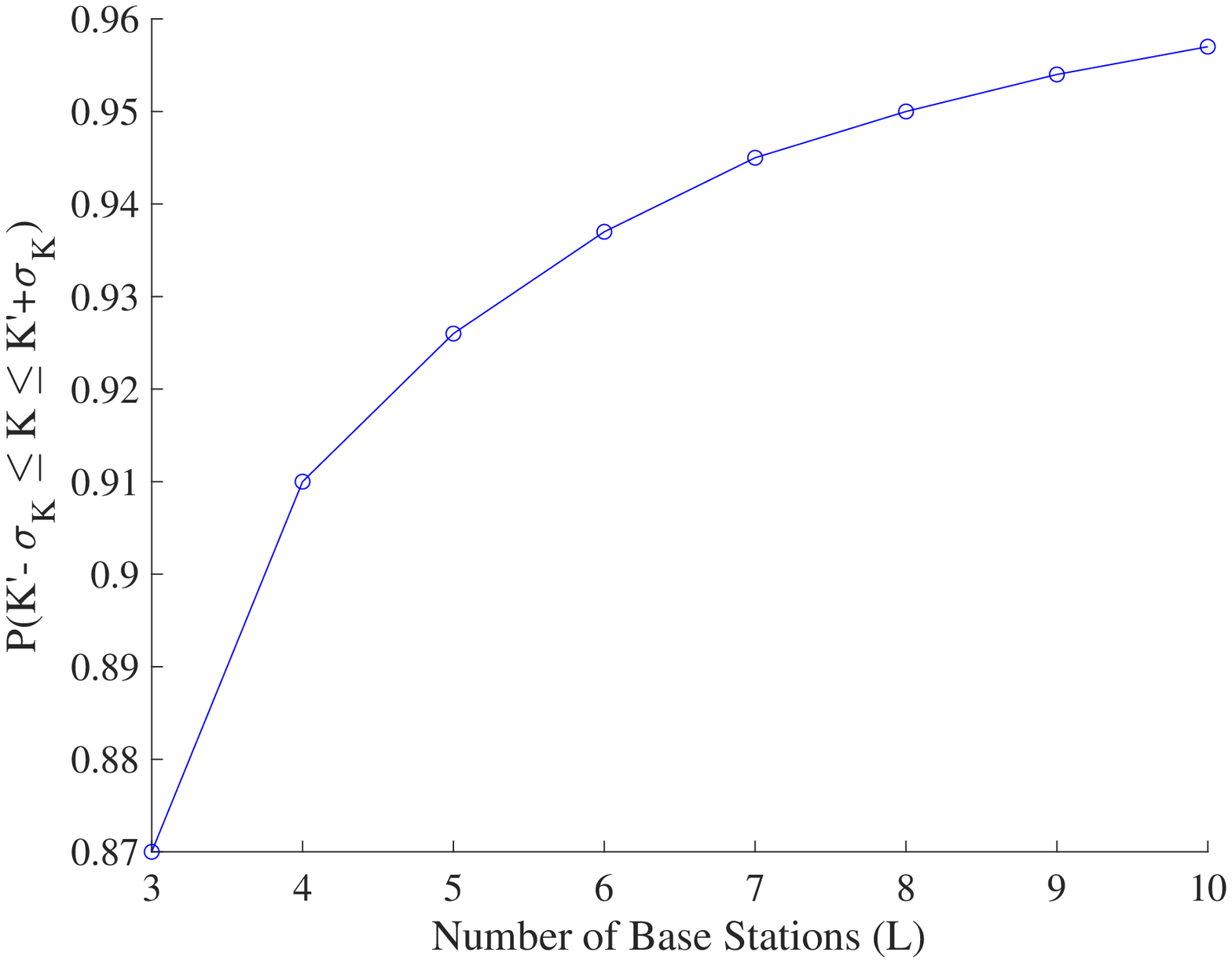}
    \caption{Probability of $K'-\sigma_{K} \leq K \leq K'+\sigma_{K}$}
    \label{fig3}
  \end{minipage}
  ~
  \begin{minipage}[b]{0.48\textwidth}
    \includegraphics[width=\columnwidth]{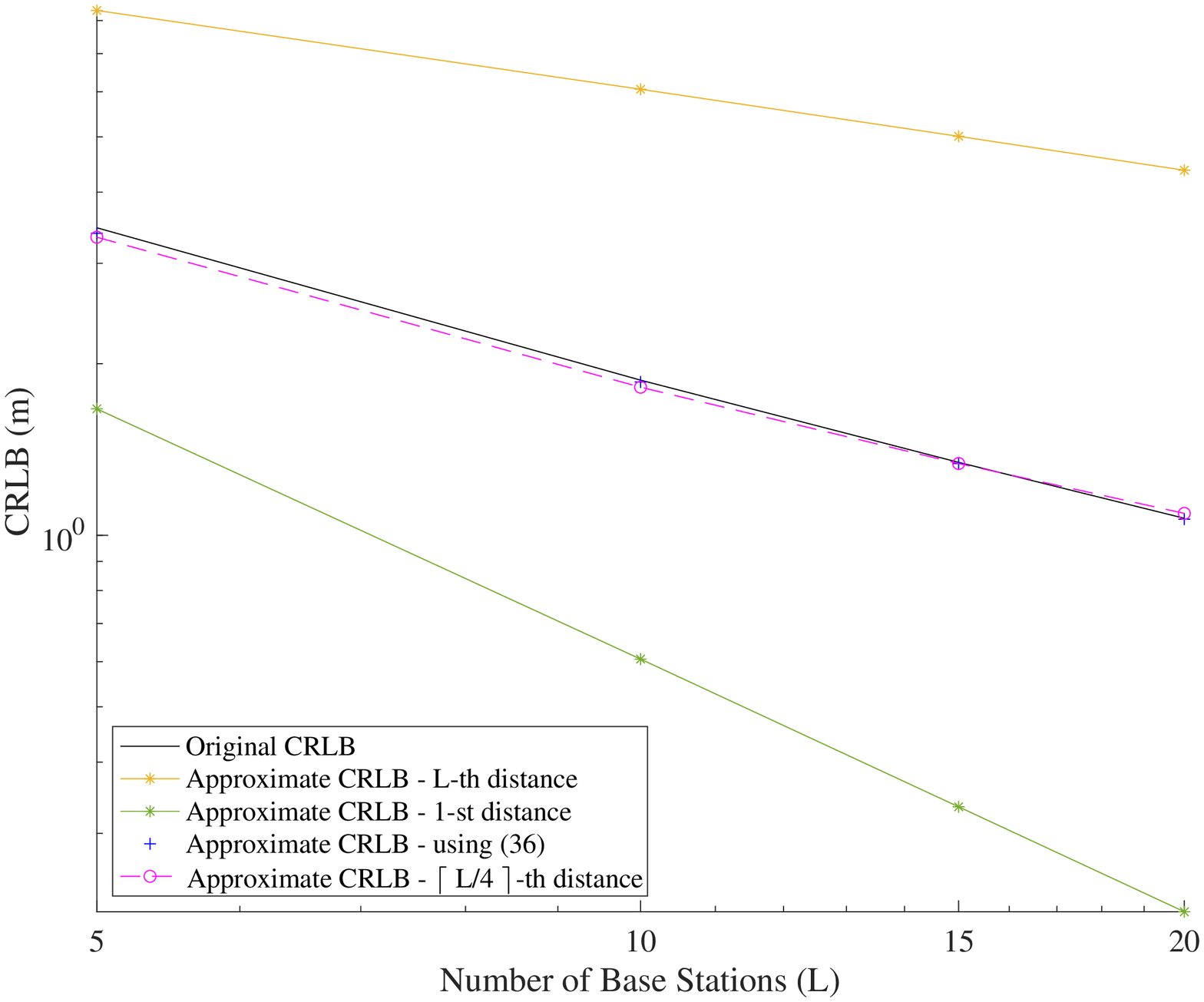}
    \caption{Original CRLB compared with approximate CRLB}
    \label{fig5}
  \end{minipage}
\end{figure*}

Based on (\ref{eq-31}) and (\ref{eq-32}), we can approximate $S$ by
\begin{equation}
S \approx \frac{2 \sigma_{\rm{AOA}}}{\sqrt{L-1}} ~r_{\lceil L/4 \rceil}^{2},
\label{eq-36}
\end{equation}
where $S$ is a function of $r_{\lceil L/4 \rceil}$. In the following proposition, we derived the distribution of $S$ using order statistic. 
\begin{prop}
Assume that the number of participating BSs $L$, the variance of the range error $\sigma_{\rm{AOA}}$, and the link distance to the closest node $r_1$ and furthest node $r_L$ are known. Then, the CDF of PEB $S$ is given by 
\begin{equation}
    \begin{split}
    F_{S}\left(s|L,\sigma_{\rm{AOA}}\right) = F_{r_{\lceil L/4 \rceil}}\left[\sqrt{\frac{s}{2}\cdot\frac{\sqrt{L-1}}{\sigma_{\rm{AOA}}}}~ \bigg|L, \sigma_{\rm{AOA}}\right], 
    \end{split}
\label{eq:s-dist}
\end{equation}
where $F_{r_n}\left(r\right)$ is the CDF of the $n$-th order statistic. 
\end{prop}

\begin{proof}
First, the PDF of the $n$-th order statistic $f_{r_{n}}(r)$ is \cite{ref32}
\begin{equation}
\begin{split}
f_{r_{n}}(r) = L f_{r_{l}}(r) \binom{L-1}{n-1} F_{r_{l}}(r)^{n-1} \left(1-F_{r_{l}}(r)\right)^{L-n},
\end{split}
\label{eq-37}
\end{equation}
where $f_{r_{l}}(r)$ and $F_{r_{l}}(r)$ are given in (\ref{eq-2}). The CDF $F_{r_n}\left(r\right)$ can be derived by integrating (\ref{eq-37}) as follows 
\begin{equation}
\begin{split}
F_{r_{n}}(r)
&= \int_{0}^{r_{L}}f_{r_{n}}(r)dr \\
&= \sum_{j=n}^{L} \binom{L}{j}
F_{r_{l}}(r)^{j}\left(1-F_{r_{l}}(r)\right)^{L-j} \\
&= \sum_{j=n}^{L} \binom{L}{j}
\left(\frac{r^2}{r_{L}^{2}-r_{1}^{2}}\right)^{j}\left(1-\frac{r^2}{r_{L}^{2}-r_{1}^{2}}\right)^{L-j}.
\end{split}
\label{eq-38}
\end{equation}

Hence, the CDF of $S$, denoted by $F_{S}(s|L, \sigma_{aoa}) = P[S\leq s|L, \sigma_{aoa}]$ is readily computed as
\begin{equation}
\begin{split}
F_{S}\left(s|L,\sigma_{\rm{AOA}}\right) &= P\left[r_{n} \leq \sqrt{\frac{s}{2}\frac{\sqrt{L-1}}{\sigma_{\rm{AOA}}}}\bigg|L, \sigma_{\rm{AOA}}\right],
\end{split}
\label{eq-39}
\end{equation}
where the PDF of $S$ can be computed by differentiating (\ref{eq-39}).
This completes the proof.
\end{proof}

\begin{remark}
We evaluated the approximation accuracy of (\ref{eq-31}) and (\ref{eq-36}) through Monte Carlo simulation. Let us denote 
\begin{equation}
    \begin{split}
        K =& \sum_{i=1}^{L}\frac{\sin^{2}\theta_{i}}{r_{i}^{2}}\sum_{j=1}^{L}\frac{\cos^{2}\theta_{j}}{r_{j}^{2}}, \quad 
        K' = \frac{1}{4}\left(\sum_{i=1}^{L}\frac{1}{r_{i}^{2}}\right)^2, 
    \end{split}
\label{eq-30}
\end{equation}
where $K$ is equal to $Q_1$ in proposition 1 and $K'$ represents the upper bound of $Q_1$ in (\ref{eq-28-e}). We utilized Monte Carlo simulation of $10$ million realizations to compute the empirical distribution of $K$ and determine the value of $P\left[K'-\sigma_{K} \leq K \leq K'+\sigma_{K}\right]$, where the $\sigma_{K}$ is the standard deviation of $K$. Fig. \ref{fig3} shows the probability $P\left[K'-\sigma_{K} \leq K \leq K'+\sigma_{K}\right] $ versus a range of $L$. It is observed that $K'$ can approximate $Q$ with high accuracy. For $L \geq 10$, the approximation accuracy is above 96$\%$. 

Furthermore, we compared the CRLB computed by using (\ref{eq-25}), (\ref{eq-31}), (\ref{eq-36}), the 1-st and the L-th ordered distance in Fig. \ref{fig5}. It is observed that the approximations of CRLB using the $L$-th and $1$-st distances cannot approach the original CRLB. However, the asymptotic bound using (\ref{eq-31}) and the approximation based on (\ref{eq-36}) closely match the original CRLB curve, which justifies Proposition 1 and Assumption 1. 
\end{remark}

\section{Simulation Results}

\begin{figure*}
  \centering
  \begin{minipage}[b]{0.48\textwidth}
    \includegraphics[width=\columnwidth]{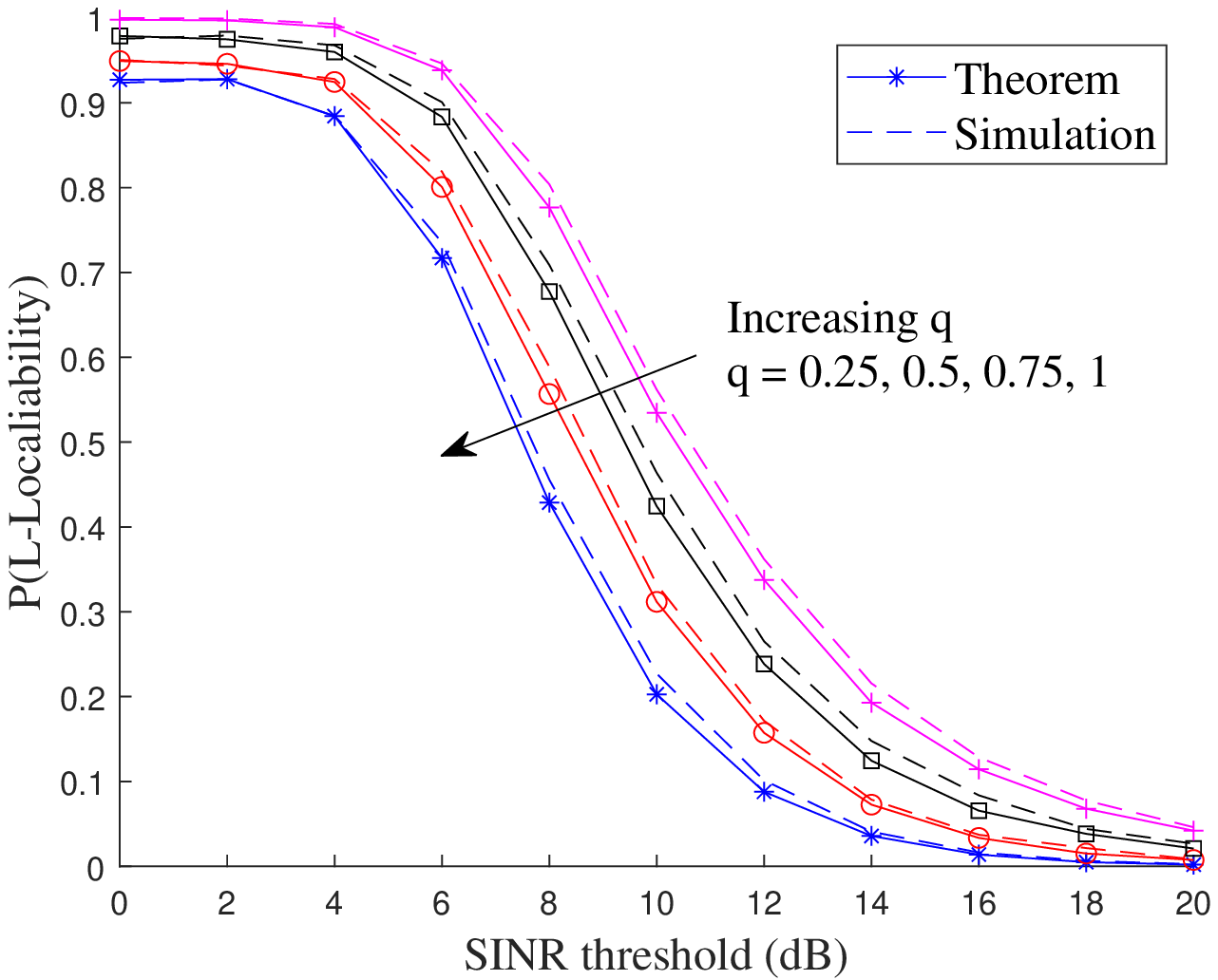}
    \caption{Impact of the network load on $L$-localizability when $\alpha = 4$, $N_t = 64$ and $M = 5$}
    \label{fig6}
  \end{minipage}
  \hfill 
  \begin{minipage}[b]{0.48\textwidth}
    \includegraphics[width=\columnwidth]{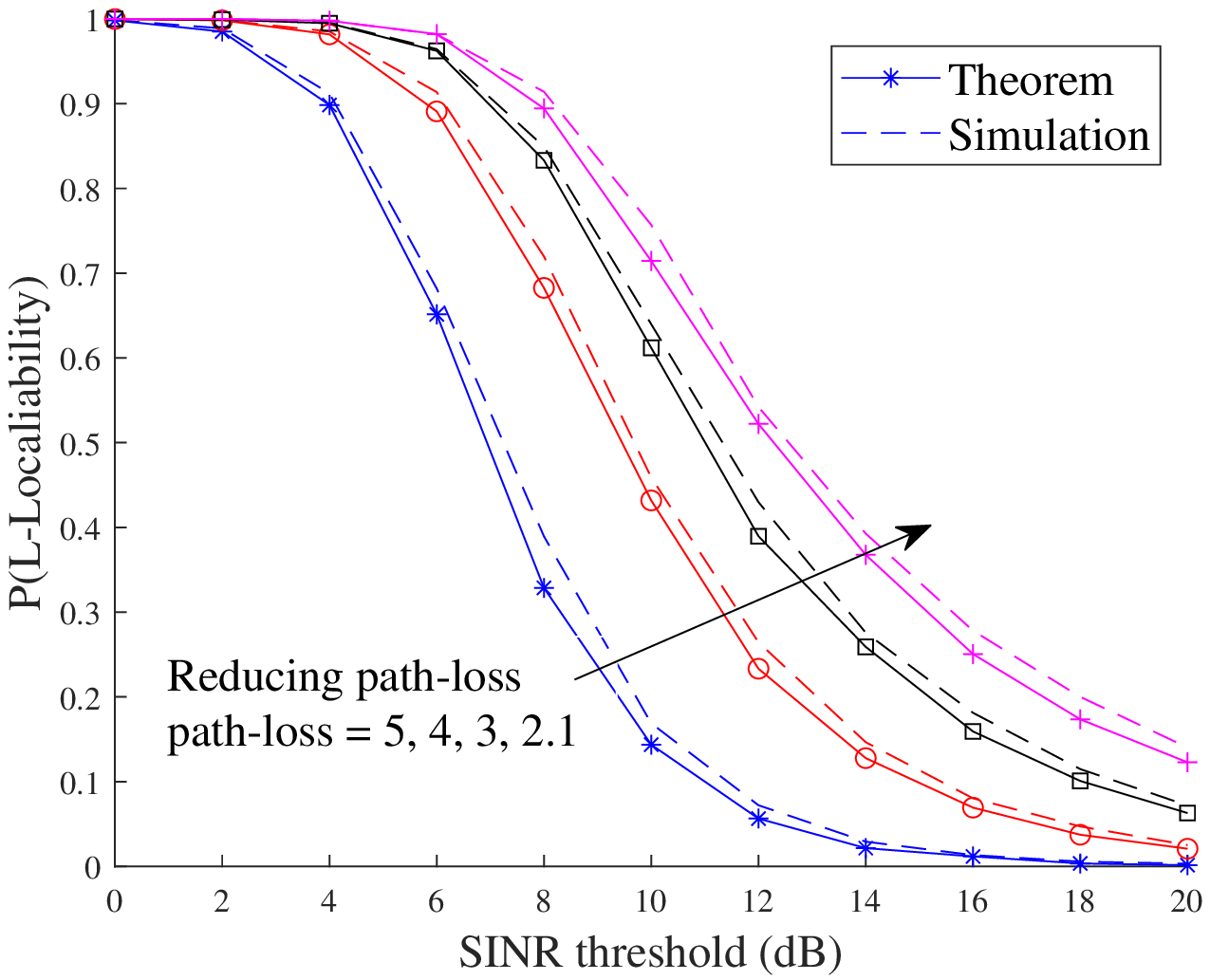}
    \caption{Impact of path-loss on $L$-localizability when $\alpha = 4$, $N_t = 64$ and $q = 0.75$}
    \label{fig7}
  \end{minipage}\\
  \begin{minipage}[b]{0.48\textwidth}
    \includegraphics[width=\columnwidth]{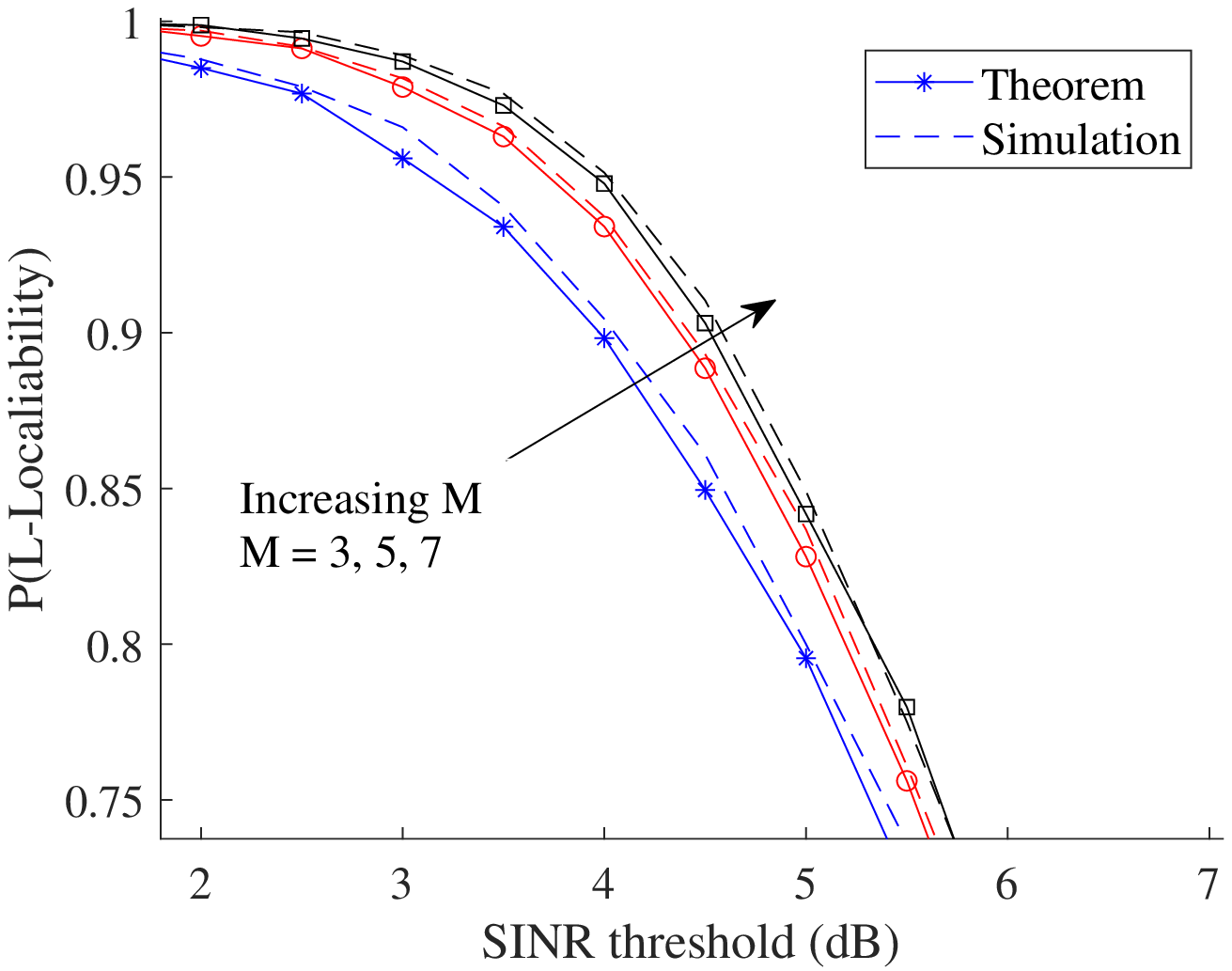}
    \caption{Impact of Nakagami fading parameter on $L$-localizability when $\alpha = 2.1$, $N_t = 64$ and $q = 0.75$}
    \label{fig8}
  \end{minipage}
  \hfill 
  \begin{minipage}[b]{0.48\textwidth}
    \includegraphics[width=\columnwidth]{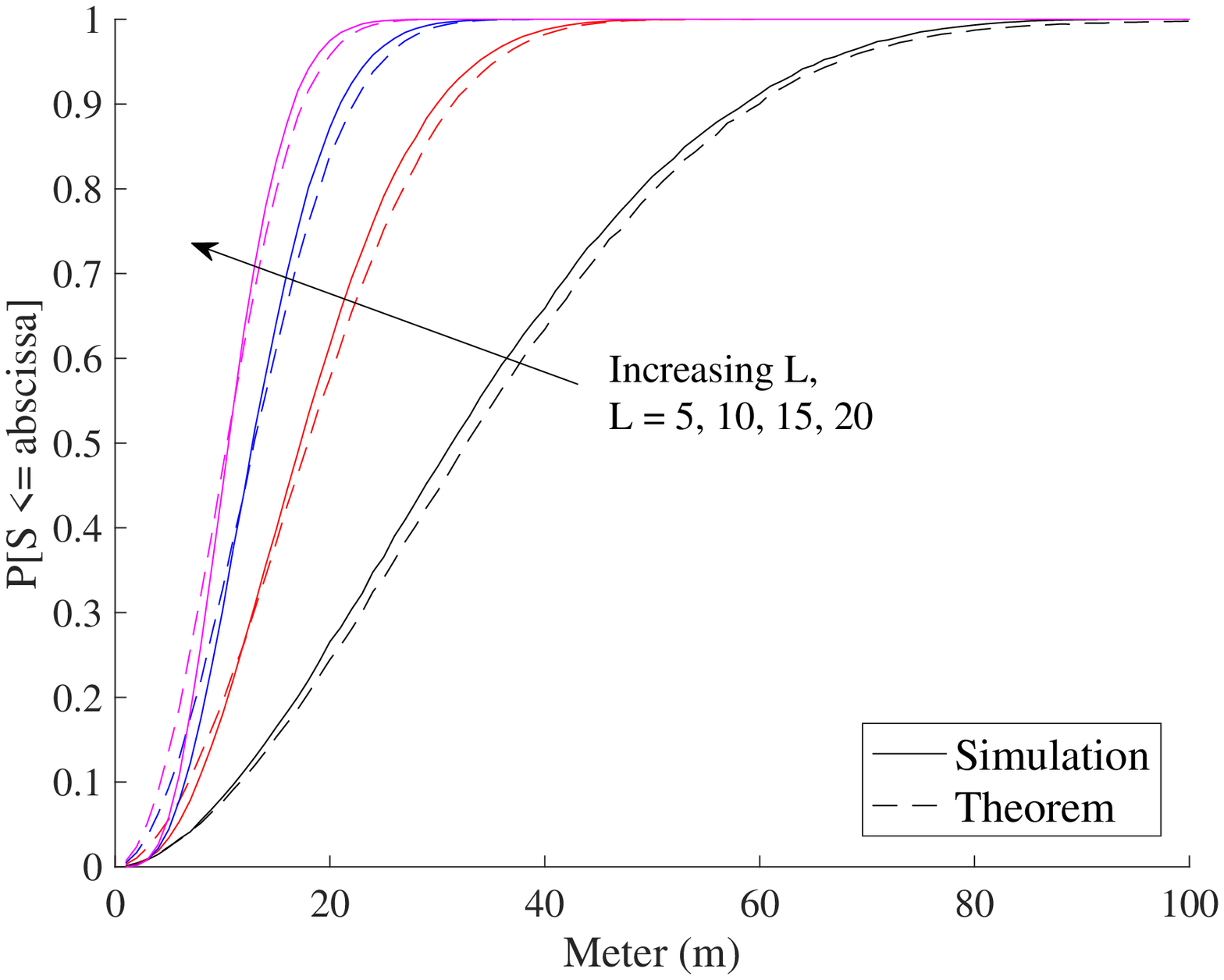}
    \caption{Impact of number of BSs on the distribution of $S$ when $\alpha = 2$, $N_t = 64$, $M = 5$ and $q = 0.75$}
    \label{fig9}
  \end{minipage}\\  
  \begin{minipage}[b]{0.48\textwidth}
    \includegraphics[width=\columnwidth]{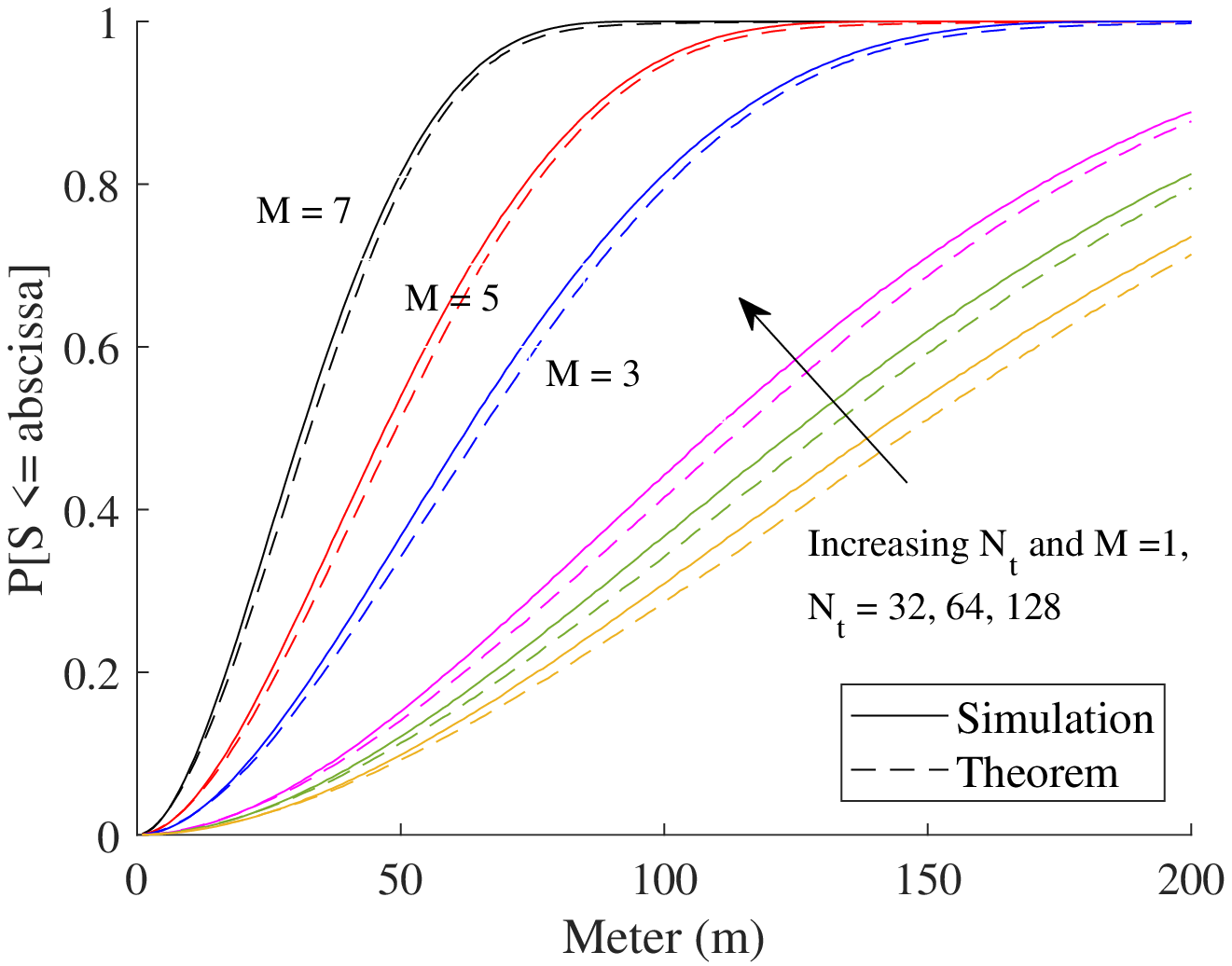}
    \caption{Impact of $M$ ($M = 3, 5, 7$) on the distribution of $S$ when $\alpha = 2$, $N_t = 64$, $L = 5$, $q = 0.75$ and the impact of $N_t$ ($N_t = 32, 64, 128$) when $\alpha = 2$, $M = 1$, $L = 5$ and $q = 0.75$}
    \label{fig10}
  \end{minipage}
  \hfill 
  \begin{minipage}[b]{0.48\textwidth}
    \includegraphics[width=\columnwidth]{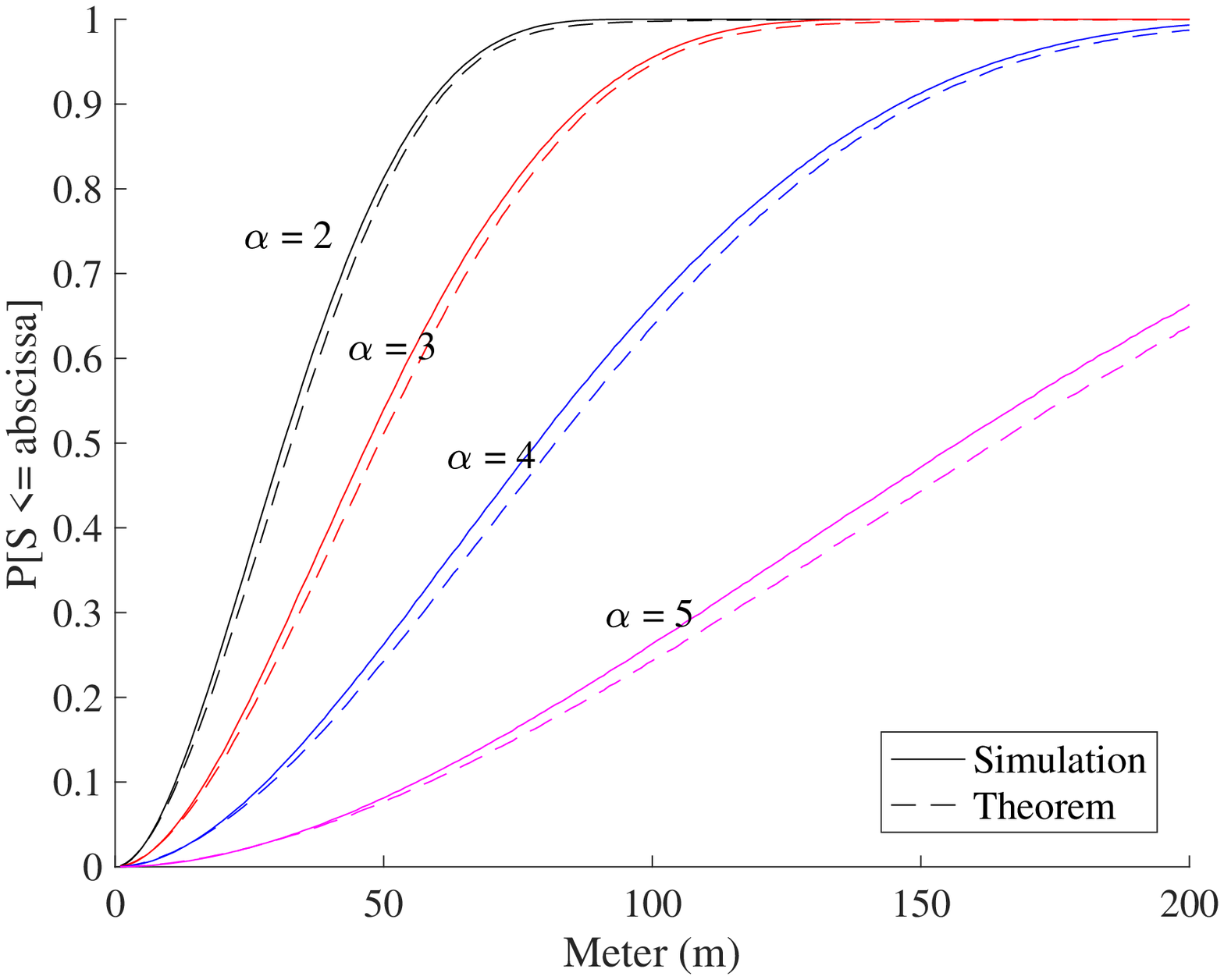}
    \caption{Impact of path-loss on the distribution of $S$ when $L = 5$, $N_t = 64$, $M = 5$ and $q = 0.75$}
    \label{fig11}
  \end{minipage}  
\end{figure*}

In this section, we evaluate the $L$-localizability and random AOA-based CRLB for mmWave networks, compare the simulation results to numerical results, and investigate the effect of network parameters on the localization performance. We used MATLAB to randomly simulate a realization of the node deployment $1 \times 10^{6}$ times. It is assumed that the BSs are randomly distributed by a homogeneous PPP with node density $\lambda = 2/\sqrt{3}\times500^{2}m^{2}$, bandwidth $W_{\rm{TOT}} = 1$ GHz, transmit power $P_T = 1$ Watt, antenna spacing  $d=\lambda_{w}/4$, path-loss intercept $\beta = \left(\lambda_{w}/4\pi\right)^{2}$, and main-lobe gain $G_{1} = 1$ and side-lobe gain $G_{2} = 0.2$ with its associate probability $p_a = 0.4$ and $p_b = 0.6$, respectively. 
Since the NLOS interference is ignored in our model, we choose a $\sigma_{\rm{AOA}}$ that accounts for an angular spread under NLOS conditions during simulation. 

\subsection{$L$-localizability Analysis}

In Figs. \ref{fig6}-\ref{fig8}, we investigate how the network parameters, including network loads, path-loss exponent and Nakagami fading parameter, affect the performance of $L$-localizability. Fig. \ref{fig6} compares the $L$-localizability $P_{L}$ versus the SINR threshold for different network loads $q$. The simulation results are plotted in dotted curves, where as the analytical results are represented by solid curves with a marker. All of the numerical results indicate that the analytical results accurately match the simulation results, justifying the analytical derivation. We observed that increasing the network load leads to a decrease in $P_{L}$. It means that network design should be optimized so that there is a sufficient number of BSs to meet the localization requirement. Fig. \ref{fig7} demonstrates the impact of path-loss on the $L$-localizability. As the path-loss exponent increases, the transmitted power across the mmWave link will significantly decline, causing a significant drop in $P_L$. In Fig. \ref{fig8}, we observe the impact of Nakagami fading parameter $M$ on $P_L$. Since the Nakagami channel becomes deterministic as the $M$ parameter increases, the $L$-localizability escalates with higher $M$ values. 

\subsection{Random AOA-based CRLB Analysis}
In Figs. \ref{fig9}-\ref{fig11}, we evaluate the distribution of $S$ for various network parameter configurations. Since the approximation of CRLB using $r_{\lceil L/4 \rceil}$ provides an accurate approximation to the original CRLB, we used the approximation based on the $\lceil L/4 \rceil$-th distance across Figs. \ref{fig9}-\ref{fig11}. In Fig. \ref{fig9}, we examine the impact of the number of participating BSs $L$ on the distribution of $S$. This is accomplished by varying the number of activated BSs transmitting during a localization procedure. It is observed that the value of $P[S \leq abscissa]$ increases for a larger $L$. 
Since the localization error reduces as the number of BSs increases, a network designer looking to improve the localization accuracy may aim to optimize the network environment to ensure a sufficient number of BSs participate in the localization procedure. 

Fig. \ref{fig10} compares the localization performance for various Nakagami fading parameter $M$ and the number of antenna elements $N_t$. It is observed that increasing $M$ parameter escalates $P[S \leq abscissa]$, which improves the localization performance. Furthermore, we demonstrate how the number of antenna elements $N_t$ affects the localization performance. As the number of antenna elements increases, the normalized noise power $\sigma_{n}^{2} = \frac{\sigma_{T}^{2}+\sigma_{{\rm{out}}}^{2}}{\beta P_{t}N_{t}}$ will be reduced, which leads to an increase of $P[S \leq abscissa]$. This indicates that the localization performance can be enhanced by adding more antenna elements in the BSs, which raises the implementation cost for each BS. Hence, the network designer should find an optimum trade-off between choosing a suitable number of antennas in the BS and enhancing the localization performance. Fig. \ref{fig11} shows the impact of path-loss on the performance of the mmWave-based localization systems. As the path-loss exponent increases, the value of $P[S \leq abscissa]$ decline, which is a similar pattern to Fig. \ref{fig7}. 

\section{Conclusion}
This paper presents $L$-localizability and random AOA-based CRLB for mmWave wireless network, where we used stochastic geometry to account for all possible positioning scenarios. We derived the $L$-localizability and random CRLB for AOA localization while considering the flat-top antenna radiation pattern and Nakagami fading. We provided numerical results to validate the analytical derivation and investigated the impact of various network parameters, \textit{e.g.}, network load, path-loss, fading parameters, number of BSs, number of antenna elements, on the localization performance. The analytical framework developed in this paper offers an accurate  tool to evaluate the localization performance of mmWave wireless networks, without relying on numerical simulation. The network operators can use the asymptotic bounds to optimize the network parameters and find the best deployment of the BSs to ensure the localization performance. In our future work, we will apply the approximation method to evaluate the performances of TOA, TDOA and RSS based localization and investigate the impact of various channel models.

\ifCLASSOPTIONcaptionsoff
  \newpage
\fi

\end{document}